\documentclass[12pt]{iopart}
\usepackage{url,amssymb,amsthm}
\usepackage{bm}
\usepackage{iopams} 
\usepackage{graphicx}
\usepackage{epstopdf}

\def\d{\rmd}
\def\e{\rme}

\def\a{\alpha}
\def\b{\beta}
\def\ep{\varepsilon}\def\ph{\varphi}
\def\k{\kappa}\def\th{\vartheta}
\def\w{\omega}
\def\tfrac#1#2{{\case{#1}{#2}}}
\def\ifrac#1#2{{#1}/{#2}}

\def\p{Painlev\'e}
\def\peq{\p\ equation}
\def\peqs{\p\ equations}
\def\bk{B\"acklund}
\def\bt{\bk\ transformation}
\def\bts{\bk\ transformations}

\def\PI{\hbox{\rm P$_{\rm I}$}}
\def\PII{\hbox{\rm P$_{\rm II}$}}
\def\PIII{\hbox{\rm P$_{\rm III}$}}
\def\PIV{\hbox{\rm P$_{\rm IV}$}}
\def\PV{\hbox{\rm P$_{\rm V}$}}
\def\PVI{\hbox{\rm P$_{\rm VI}$}}
\def\PIIIi{\mbox{\rm P$_{\rm III'}$}}

\def\dPIV{\hbox{\rm dP$_{\rm IV}$}}

\def\HV{\hbox{$\mathcal{H}_{\rm V}$}}

\def\HIIIi{\hbox{$\mathcal{H}_{\rm III'}$}}

\def\sIII{\sigma}
\def\sV{\sigma}

\def\bk{B\"acklund}
\def\hide#1{}

\def\W{\mathcal{W}}
\def\ds{\displaystyle}
\newcommand{\deriv}[3][]{\frac{\d^{#1}{#2}}{{\d{#3}}^{#1}}}
\def\dz#1{\delta_{#1}}

\def\Z{\mathbb{Z}}
\def\N{\mathbb{N}}
\def\R{\mathbb{R}}

\newcommand{\BesselJ}[1]{J_{#1}}

\newcommand{\BesselI}[1]{I_{#1}}
\newcommand{\BesselK}[1]{K_{#1}}

\newcommand{\WhitM}[2]{M_{#1,#2}}
\newcommand{\WhitW}[2]{W_{#1,#2}}

\newcommand{\HyperpFq}[2]{{}_{#1}F_{#2}}

\newtheorem{theorem}{Theorem}[section]

\newtheorem{lemma}[theorem]{Lemma}
\newtheorem{corollary}[theorem]{Corollary}
\theoremstyle{definition}

\newtheorem{remarks}[theorem]{Remarks}
\def\intS{\int_a^b}

\newcommand{\pderiv}[3][]{\frac{\partial^{#1}{#2}}{{\partial{#3}}^{#1}}}

\eqnobysec
\begin{document}

\title[Discrete orthonormal polynomials and the \p\ equations]{Recurrence coefficients for discrete orthonormal polynomials  and the \p\ equations}
\author{Peter A.\ Clarkson}
\address{School of Mathematics, Statistics and Actuarial Science,
University of Kent, Canterbury, CT2 7NF, UK}
 \ead{P.A.Clarkson@kent.ac.uk}

\begin{abstract}
We investigate semi-classical generalizations of the Charlier and Meixner polynomials, which are discrete orthogonal polynomials that satisfy three-term recurrence relations. It is shown that the coefficients in these recurrence relations can be expressed in terms of Wronskians of modified Bessel functions and confluent hypergeometric functions, respectively for the generalized Charlier and generalized Meixner polynomials. These Wronskians arise in the description of special function solutions of the third and fifth \p\ equations.
\end{abstract}

\ams{34M55, 33E17, 33C47}
\maketitle

\section{Introduction}
In this paper we are concerned with the coefficients in the three-term recurrence relations for semi-classical orthonormal polynomials, specifically for  generalizations of the Charlier and Meixner polynomials which are discrete orthonormal polynomials. It is shown that these recurrence coefficients for the generalized Charlier polynomials and generalized Meixner polynomials can respectively be expressed in terms of Wronskians that arise in the description of special function solutions of the third \p\ equation (\PIII)
\begin{equation}\label{eq:PT.DE.PIII}
\deriv[2]{w}{z} = \frac1w \left(\deriv{w}{z}\right)^{2}- \frac{1}{z} \deriv{w}{z}+\frac{Aw^2 + B}{z} + C w^3 + \frac{D}{w},
\end{equation} 
where $A$, $B$, $C$ and $D$ are arbitrary constants, and the fifth \p\ equation (\PV)
\begin{eqnarray}
\deriv[2]{w}{z} &= \left(\frac{1}{2w} + \frac{1}{w-1}\right)\! \left(\deriv{w}{z} \right)^{2} - \frac{1}{z} \deriv{w}{z} 
+\frac{(w-1)^2}{z^2}\left(A w +\frac{B}{w}\right)\nonumber\\ &\qquad+ \frac{C w}{z} + \frac{D w(w+1)}{w-1}.
\label{eq:PT.DE.PV}\end{eqnarray} 
Wronskians for special function solutions of \PIII\ are expressed in terms of modified Bessel functions and for \PV\  in terms of confluent hypergeometric functions. 

The relationship between semi-classical orthogonal polynomials and integrable equations dates back to the work of Shohat \cite{refShohat} in 1939 and later Freud \cite{refFreud} in 1976. However it was not until the work of Fokas, Its and Kitaev \cite{refFIKa,refFIKb} in the early 1990s that these equations were identified as discrete \p\ equations. The relationship between semi-classical orthogonal polynomials and the (continuous) \p\ equations was demonstrated by Magnus \cite{refMagnus95} in 1995. 
\hide{who studied the coefficients in the three-term recurrence relationship for the Freud weight \cite{refFreud}
\begin{equation*}\omega(x;t)=\exp\left(-\tfrac14x^4-tx^2\right),\qquad x\in\R,\end{equation*} with $t\in\R$ a parameter.}
A motivation for this work is that recently 
it has been shown that recurrence coefficients for several semi-classical orthogonal polynomials can be expressed in terms of solutions of \p\ equations, see, for example, \cite{refBC09,refBCE10,refBFvA,refBFSvAZ,refBvA,refCD10,refCF06,refCI10,refCP05,refCZ10,refPACJordaan13,refDZ10,%
refFvA11,refFvA13,refFvAZ,refFO10,refFW07,refSmetvA,refvA07,refvAF03}. 

This paper is organized as follows: 
in \S\ref{sec:peq} we review properties of \PIII\ (\ref{eq:PT.DE.PIII}) and \PV\ (\ref{eq:PT.DE.PV}), including special function solutions and the Hamiltonian structure of these equations;
in \S\ref{sec:op} we review properties of orthogonal polynomials and discrete orthogonal polynomials;
in \S\ref{sec:char} we derive expressions for the recurrence coefficients for the generalized Charlier polynomials in terms of Wronskians that arise in the description of special function solutions of \PIII;
in \S\ref{sec:meix} we derive expressions for the recurrence coefficients for the generalized Meixner polynomials in terms of Wronskians that arise in the description of special function solutions of \PV; and
in \S\ref{sec:dis} we discuss our results.

\section{\label{sec:peq}\peqs}
The six \peqs\ (\PI--\PVI) were first discovered by \p, Gambier and their colleagues in an investigation of which second order ordinary differential equations of the form 
\begin{equation} \label{eq:PT.INT.gen-ode} \deriv[2]{w}{z}=F\left(\deriv{w}{z},w,z\right),  \end{equation} 
where $F$ is rational in $\d w/\d z$ and $w$ and analytic in $z$, have the property that their solutions have no movable branch points. 
They showed that there were fifty canonical equations of the form (\ref{eq:PT.INT.gen-ode}) with this property, now known as the \textit{\p\ property}. Further \p, Gambier and their colleagues showed that of these fifty equations, forty-four can be reduced to linear equations, solved in terms of elliptic functions, or are reducible to one of six new nonlinear ordinary differential equations that define new transcendental functions (see Ince \cite{refInce}). \hide{Amongst these six new nonlinear ordinary differential equations are the third \p\ equation (\PIII)
\begin{equation}\label{eq:PT.DE.PIII}
\deriv[2]{w}{z} = \frac1w \left(\deriv{w}{z}\right)^{2}- \frac{1}{z} \deriv{w}{z}
 +\frac{Aw^2 + B}{z} + C w^3 + \frac{D}{w},
\end{equation}
and the fifth \p\ equation (\PV)
\begin{eqnarray}
\deriv[2]{w}{z} &= \left(\frac{1}{2w} + \frac{1}{w-1}\right)\!
\left(\deriv{w}{z} \right)^{2} - \frac{1}{z} \deriv{w}{z} + 
\frac{(w-1)^2}{z^2}\left(A w +\frac{B}{w}\right)\nonumber\\ &\qquad+ 
\frac{C w}{z} + \frac{D w(w+1)}{w-1},\label{eq:PT.DE.PV}\end{eqnarray}
where $A$, $B$, $C$ and $D$ are arbitrary constants.}%
The \p\ equations can be thought of as nonlinear analogues of the classical special functions \cite{refPAC05review,refFIKN,refGLS02,refIKSY}, and arise in a wide variety of applications, for example random matrices, see \cite{refForrester,refOsipovKanz} and the references therein.

The \peqs\ \PII--\PVI\ possess hierarchies of solutions expressible in terms of classical special functions, 
cf.~\cite{refPAC05review,refGLS02,refMasuda04} and the references therein. 
For \PIII\ (\ref{eq:PT.DE.PIII}) these are expressed in terms of Bessel functions
\cite{refMCB97,refMurata95,refOkamotoPIII}, which we discuss in \S\ref{ssec:piiisols},
and for \PV\ (\ref{eq:PT.DE.PV}) in terms of confluent hypergeometric functions (equivalently, Kummer functions or Whittaker functions)
\cite{refOkamotoPV,refMasuda04,refWat}, which we discuss in \S\ref{ssec:pvsols}.

Each of the \peqs\ \PI--\PVI\ can be written as a (non-autonomous) Hamiltonian system. For \PIII\ (\ref{eq:PT.DE.PIII}) and  \PV\ (\ref{eq:PT.DE.PV}) these have the form
\begin{equation*}\label{sec:PT.HM.DE1}
z\deriv{q}{z}=\pderiv{\mathcal{H}_{\rm J}}{p},\qquad 
z\deriv{p}{z}=-\pderiv{\mathcal{H}_{\rm J}}{q},\qquad {\rm J}={\rm III, V}
\end{equation*}
for a suitable Hamiltonian function $\mathcal{H}_{\rm J}(q,p,z)$\ \cite{refJM,refOkamoto80a,refOkamotoPV,refOkamotoPIII}, which we discuss in \S\ref{ssec:piiiham} and \S\ref{ssec:pvham}. Further, the function $\sigma(z)\equiv\mathcal{H}_{\rm J}(q,p,z)$ satisfies a second-order, second-degree equation, which is often called the \textit{Jimbo-Miwa-Okamoto equation} or \textit{\p\ $\sigma$-equation}, whose solution is expressible in terms of the solution of the associated \peq\ \cite{refJM,refOkamoto80b,refOkamotoPV,refOkamotoPIII}. Hence there are special function solutions of these equations, which are also discussed in \S\ref{ssec:piiiham} and \S\ref{ssec:pvham}. 

\subsection{\label{ssec:piiisols}Special functions solutions of the third \peq.}
In the generic case when $CD\not=0$ in \PIII\ (\ref{eq:PT.DE.PIII}), then we set $C=1$ and $D=-1$, without loss of generality, so in the sequel we consider the equation 
\begin{equation}\label{eq:PT.DE.PIIIgen}
\deriv[2]{w}{z} = \frac1w \left(\deriv{w}{z}\right)^{2}- \frac{1}{z} \deriv{w}{z}
 +\frac{Aw^2 + B}{z} + w^3 - \frac{1}{w}.\end{equation}
Special function solutions of (\ref{eq:PT.DE.PIIIgen}) are expressed in terms of Bessel functions, see \cite{refMCB97,refMurata95,refOkamotoPIII}. 

\begin{theorem}\label{thm:P3sf}{\Eref{eq:PT.DE.PIIIgen} has solutions expressible in terms of Bessel functions if and only if
\begin{equation}\label{eq:PT.SF.eq30} \ep_1A+\ep_2B= 4n+2, \end{equation} 
with $n\in\Z$ and $\ep_1=\pm1$, $\ep_2 =\pm1$ independently.}\end{theorem} 

\begin{proof}{See Gromak \cite{refGromak78}, Mansfield and Webster \cite{refMW} and Umemura and Watanabe \cite{refUW98}; also \cite[\S35]{refGLS02}.}\end{proof}

\hide{In the case $\ep_1A+\ep_2B=2$, setting $A=2\nu$, then (\ref{eq:PT.DE.PIIIgen}) has solutions expressible in terms of solutions of the Riccati equation
\begin{equation*}\label{eq:PT.SF.eq31}
z\deriv{w}z=\ep_1zw^2+(2\ep_1\nu-1)w+\ep_2z,
\end{equation*} 
which has the solution 
\begin{equation*}\label{eq:PT.SF.eq32}
w(z)=-\ep_1\,\deriv{}{z}\ln\ph_{\nu}(z),
\end{equation*} 
where $\ph_{\nu}(z)$ satisfies
\begin{equation*}\label{eq:PT.SF.eq33}z\deriv[2]{\ph_{\nu}}{z}+(1-2\ep_1\nu)\deriv{\ph_{\nu}}{z}+\ep_1\ep_2z\ph_{\nu}=0,\end{equation*}
which has solution
This equation has the solution 
\begin{equation*}\label{eq:PT.SF.eq33sol}
\ph_{\nu}(z)=\cases{ z^{\nu}\left\{C_1J_{\nu}(z)+C_2Y_{\nu}(z)\right\},&$\mbox{\rm if}\quad \ep_1=1,\phantom{-}\enskip\ep_2=1,$\\
z^{-\nu}\left\{C_1J_{\nu}(z)+C_2Y_{\nu}(z)\right\},&$\mbox{\rm if}\quad \ep_1=-1,\enskip\ep_2=-1,$\\
z^{\nu}\left\{C_1I_{\nu}(z)+C_2K_{\nu}(z)\right\},&$\mbox{\rm if}\quad\ep_1=1,\phantom{-}\enskip\ep_2=-1,$\\
z^{-\nu}\left\{C_1I_{\nu}(z)+C_2K_{\nu}(z)\right\},&$\mbox{\rm if}\quad\ep_1=-1,\enskip\ep_2=1,$}\end{equation*}
with $C_1$ and $C_2$ arbitrary constants, and where $J_{\nu}(z)$, $Y_{\nu}(z)$, $I_{\nu}(z)$ and $K_{\nu}(z)$ are Bessel functions.}%

An alternative form of \PIII, due to Okamoto \cite{refOkamoto80a,refOkamoto80b,refOkamotoPIII}, is obtained by making the transformation $w(z)=u(t)/\sqrt{t}$, with $t=\tfrac14z^2$, in (\ref{eq:PT.DE.PIIIgen}) giving
\begin{equation}\label{eq:PT.DE.PIIIi}
\deriv[2]{u}{t}=\frac1u\left(\deriv{u}{t}\right)^2-\frac1t \,\deriv{u}{t} +\frac{Au^2}{2t^2}+\frac{B}{2t}+\frac{u^3}{t^2}-\frac1u,
\end{equation} 
which is known as \PIIIi.
\Eref{eq:PT.DE.PIIIi} has solutions expressible in terms of solutions of the Riccati equation
\begin{equation}\label{eq:PT.SF.eq31i} 
t\deriv{u}{t}=\ep_1u^2+\nu u +\ep_2t,
\end{equation} 
if and only if $A$ and $B$ satisfy (\ref{eq:PT.SF.eq30}).
\hide{\begin{equation*}\label{eq:PT.SF.eq30i} 
\ep_1A+\ep_2B=4n+2,
\end{equation*} 
with $n\in\Z$ and $\ep_1=\pm1$, $\ep_2=\pm1$, independently.}%
To solve (\ref{eq:PT.SF.eq31i}), we make the transformation
\begin{equation*}\label{eq:PT.SF.eq32i} 
u(t)=-\ep_1t\deriv{}{t}\ln\psi_{\nu}(t),
\end{equation*} 
then $\psi_{\nu}(t)$ satisfies
\begin{equation}\label{eq:PT.SF.eq33i} 
t\deriv[2]{\psi_{\nu}}{t}+(1-\nu)\deriv{\psi_{\nu}}{t}+\ep_1\ep_2\psi_{\nu}=0,
\end{equation} 
which has solution
\begin{equation}\label{eq:PT.SF.eq33isol} \fl \psi_{\nu}(t)=
\cases{ 
t^{\nu/2}\left\{C_1J_{\nu}\big(2\sqrt{t}\big)+C_2Y_{\nu}\big(2\sqrt{t}\big)\right\},&$\mbox{\rm if}\quad \ep_1=1,\phantom{-}\enskip\ep_2=1,$\\
t^{-\nu/2}\left\{C_1J_{\nu}\big(2\sqrt{t}\big)+C_2Y_{\nu}\big(2\sqrt{t}\big)\right\},&$\mbox{\rm if}\quad \ep_1=-1,\enskip\ep_2=-1,$\\
t^{\nu/2}\left\{C_1I_{\nu}\big(2\sqrt{t}\big)+C_2K_{\nu}\big(2\sqrt{t}\big)\right\},&$\mbox{\rm if}\quad\ep_1=1,\phantom{-}\enskip\ep_2=-1,$\\
t^{-\nu/2}\left\{C_1I_{\nu}\big(2\sqrt{t}\big)+C_2K_{\nu}\big(2\sqrt{t}\big)\right\},&$\mbox{\rm if}\quad\ep_1=-1,\enskip\ep_2=1,$}\end{equation}
with $C_1$ and $C_2$ arbitrary constants, and where $J_{\nu}(z)$, $Y_{\nu}(z)$, $I_{\nu}(z)$ and $K_{\nu}(z)$ are {Bessel functions}.

\subsection{\label{ssec:piiiham}Hamiltonian structure for the third \peq.}
The Hamiltonian associated with \PIIIi\ (\ref{eq:PT.DE.PIIIi}) is 
\begin{equation}\label{eq:PT.HM3p.DE31}
\HIIIi(q,p,t) = q^2p^2 -\left(q^2+\th_0q-t\right)p +\tfrac12(\th_0+\th_{\infty})q,
\end{equation} with $\th_0$ and $\th_{\infty}$ parameters, 
where $p$ and $q$ satisfy
\begin{equation}\label{eq:PT.HM3p.DE323}\begin{array}{l}
\ds t\deriv{q}{t}=2q^2p-q^2-\th_0q+t,\\[7.5pt]
\ds t\deriv{p}{t}=-2qp^2+2qp+\th_0p-\tfrac12(\th_0+\th_{\infty}),
\end{array}\end{equation} 
see Okamoto \cite{refOkamoto80a,refOkamoto80b,refOkamotoPIII}.
Eliminating $p$ in (\ref{eq:PT.HM3p.DE323}) then $q=u$ satisfies \PIIIi\ (\ref{eq:PT.DE.PIIIi}) with parameters  $(A,B)=(-2\th_{\infty},2(\th_0+1))$.
Eliminating $q$ in (\ref{eq:PT.HM3p.DE323}) then $p$ satisfies
\begin{eqnarray}\deriv[2]{p}t  &= \frac12\left(\frac{1}{p}+ \frac{1}{p-1}\right)  \left(\deriv{p}t\right) ^{2}-\frac1t\deriv{p}t-\frac {2p(p-1)}{t}
\nonumber\\&\qquad
+\frac{1}{8t^2}\left\{4\th_0\th_{\infty} -\frac{(\th_0+\th_{\infty})^2}{p}-\frac{(\th_0-\th_{\infty})^2}{p-1}\right\}.\label{eq:PT.HM3p.DE33a}
\end{eqnarray}
Making the transformation $p(t)=1/[1-w(z)]$, with $z=t$, 
in (\ref{eq:PT.HM3p.DE33a}) yields 
\begin{eqnarray}
\deriv[2]{w}{z} &= \left(\frac{1}{2w} + \frac{1}{w-1}\right)\! \left(\deriv{w}{z} \right)^{2} - \frac{1}{z} \deriv{w}{z} \nonumber\\&\qquad
+\frac{(w-1)^2}{z^2}\left\{\frac{(\th_0+\th_{\infty})^2w}8 -\frac{(\th_0-\th_{\infty})^2}{8w}\right\}- \frac{2 w}{z},
\label{eq:PT.DE.PV0}\end{eqnarray} 
which is \PV\ (\ref{eq:PT.DE.PV}) with parameters
$(A,B,C,D)=\big(\tfrac18(\th_0+\th_{\infty})^2, -\tfrac18(\th_0-\th_{\infty})^2,-2,0\big)$,
illustrating the well-known relationship between \PIII\ and \PV\ with $\delta=0$, cf.~\cite[\S34]{refGLS02}. 
\hide{Alternatively, making the transformation $p(t)=1/[1-w(z)]$, with $z=2\sqrt{t}$, in (\ref{eq:PT.HM3p.DE33a}) yields 
\begin{eqnarray}
\deriv[2]{w}{z} &= \left(\frac{1}{2w} + \frac{1}{w-1}\right)\! \left(\deriv{w}{z} \right)^{2} - \frac{1}{z} \deriv{w}{z} \nonumber\\&\qquad
+\frac{(w-1)^2}{2z^2w}\left\{(\th_0+\th_{\infty})^2w^2 -(\th_0-\th_{\infty})^2\right\}-2w,
\label{eq:PT.DE.PVa}\end{eqnarray} 
which is \PV\ (\ref{eq:PT.DE.PV}) with 
$(\a,\b,\gamma,\delta)=\big(\tfrac12(\th_0+\th_{\infty})^2, -\tfrac12(\th_0-\th_{\infty})^2,0,-2\big)$.}

The second-order, second-degree equation satisfied by the Hamiltonian function is given in the following theorem.

\begin{theorem}{The Hamiltonian function
\begin{equation*}\label{eq:PT.HM3p.DE33b}
\sIII(t;\th_0,\th_{\infty})=t\HIIIi(q,p,t)-\tfrac12t+\tfrac14\th_0^2,
\end{equation*} 
with $\HIIIi(q,p,t)$ given by (\ref{eq:PT.HM3p.DE31}), satisfies 
the second-order, second-degree equation 
\begin{equation}\label{eq:PT.HM3p.DE34}\fl
\left(t\deriv[2]{\sIII}{t}\right)^{2}+\left\{4\left(\deriv{\sIII}{t}\right)^{2}-1\right\}
\left(t\deriv{\sIII}{t}-\sIII\right) +\th_0\th_{\infty}\deriv{\sIII}{t}=\tfrac14(\th_0^2+\th_{\infty}^2).
\end{equation} 
Conversely, if $\sIII(t)$ satisfies (\ref{eq:PT.HM3p.DE34}) then the solution of the Hamiltonian system 
(\ref{eq:PT.HM3p.DE323}) is given by
\begin{equation}\label{eq:PT.HM3p.DE35}
q(t)=2\frac{\ds t\deriv[2]{\sIII}{t}-2\th_0\deriv{\sIII}{t} +\th_{\infty}}{\ds 1-4\left(\deriv{\sIII}{t}\right)^2},\qquad p(t)=\deriv{\sIII}{t}+\tfrac12. 
\end{equation}}\end{theorem}
\begin{proof}{See Okamoto \cite{refOkamoto80b,refOkamotoPIII}; see also Forrester and Witte \cite{refFW02}.}\end{proof}

The special function solutions of (\ref{eq:PT.HM3p.DE34}) are given in the following theorem.

\begin{theorem}{Let $\tau_{n,\nu}(t)$ 
be the determinant given by 
\begin{equation}\label{def:Wn}\tau_{n,\nu}(t)=\left|\matrix{ \psi_{\nu} & \delta_{t}(\psi_{\nu}) & \cdots & \delta_{t}^{n-1}(\psi_{\nu})\cr
\delta_{t}(\psi_{\nu}) & \delta_{t}^2(\psi_{\nu}) & \cdots & \delta_{t}^{n}(\psi_{\nu})\cr
\vdots & \vdots & \ddots & \vdots \cr
\delta_{t}^{n-1}(\psi_{\nu}) & \delta_{t}^{n}(\psi_{\nu}) & \cdots & \delta_{t}^{2n-2}(\psi_{\nu})
}\right|,\qquad \delta_{t}\equiv t\deriv{}{t}, \end{equation}
with $\psi_{\nu}(t)$ a solution of (\ref{eq:PT.SF.eq33i}).
Then special function solutions of (\ref{eq:PT.HM3p.DE34}) are given by
\begin{equation}\label{sol:SIIIi}
\sigma(t)= 
t\deriv{}{t}\big(\ln\tau_{n,\nu}(t)\big) +\tfrac12\ep_1\ep_2t+\tfrac14{\nu}^{2}+\tfrac12n(1-\ep_1\nu)-\tfrac14{n}^{2},\end{equation} 
for the parameters $(\th_0,\th_{\infty})=(\nu+n,\ep_1\ep_2(\nu-n))$.}\end{theorem}
\begin{proof}{See Okamoto \cite{refOkamotoPIII}; see also Forrester and Witte \cite{refFW02}.}\end{proof}

The determinant $\tau_{n,\nu}(t)$ given by (\ref{def:Wn}) is often called a ``$\tau$-function", see \cite{refOkamotoPIII}.
\hide{Further $\tau_{n,\nu}(t)$ given by (\ref{def:Wn}), where $\psi_{\nu}(t)$ is given by (\ref{eq:PT.SF.eq33isol}) with $C_1=0$ or $C_2=0$, can be written as a Toeplitz determinant.
\begin{lemma}{If $\W_n(\psi)$ is defined by 
\[\W_n(\psi)=\left|\matrix{ \psi & \delta_{t}(\psi) & \cdots & \delta_{t}^{n-1}(\psi)\cr
\delta_{t}(\psi) & \delta_{t}^2(\psi) & \cdots & \delta_{t}^{n-1}(\psi)\cr
\vdots & \vdots & \ddots & \vdots \cr
\delta_{t}^{n-1}(\psi) & \delta_{t}^{n}(\psi) & \cdots & \delta_{t}^{2n-2}(\psi)
}\right|,\qquad \delta_{t}\equiv t\deriv{}{t}, \] then
 \begin{eqnarray*}
  \W_n\Big({t}^{\pm\nu/2}\BesselJ{\nu}\big(2\sqrt{t}\big)\Big)&=(-t)^{n(n-1)/2} 
 \det\Big[t^{\pm\nu/2}\BesselJ{\nu+j-k}\big(2\sqrt{t}\big)\Big]_{j,k=0}^{n-1},\\
 \W_n\Big({t}^{\pm\nu/2}\BesselI{\nu}\big(2\sqrt{t}\big)\Big)&={t}^{n(n-1)/2}
 \det\Big[t^{\pm\nu/2}\BesselI{\nu+j-k}\big(2\sqrt{t}\big)\Big]_{j,k=0}^{n-1},\\
 \W_n\Big({t}^{\pm\nu/2}\BesselK{\nu}\big(2\sqrt{t}\big)\Big)&={t}^{n(n-1)/2}
 \det\Big[t^{\pm\nu/2}\BesselK{\nu+j-k}\big(2\sqrt{t}\big)\Big]_{j,k=0}^{n-1}.
 \end{eqnarray*}}\end{lemma}
\begin{proof}{These are easily proved using Bessel function identities, cf.~\cite[\S10.6]{refDLMF}.}\end{proof}}%

\subsection{\label{ssec:pvsols}Special functions solutions of the fifth \peq.}
In the generic case when $D\not=0$ in \PV\ (\ref{eq:PT.DE.PV}), then we set $D=-\tfrac12$, without loss of generality, so in the sequel we consider the equation  
\begin{eqnarray}
\deriv[2]{w}{z} &= \left(\frac{1}{2w} + \frac{1}{w-1}\right)\!
\left(\deriv{w}{z} \right)^{2} - \frac{1}{z} \deriv{w}{z} + 
\frac{(w-1)^2}{z^2}\left(A w +\frac{B}{w}\right)\nonumber\\ &\qquad+ 
\frac{C w}{z} - \frac{w(w+1)}{2(w-1)}.\label{eq:PT.DE.PVgen}\end{eqnarray}
Special function solutions of (\ref{eq:PT.DE.PVgen}) are expressed in terms of confluent hypergeometric functions (equivalently Kummer functions or Whittaker functions), see \cite{refOkamotoPV,refMasuda04,refWat}. 

\begin{theorem}\label{thm:p5sf} \Eref{eq:PT.DE.PVgen} then has solutions expressible in terms of Kummer functions if and only if 
\begin{equation} \label{eq:PT.SF.eq50a} a+b+\ep_3C=2N+1,
\end{equation} or
\begin{equation} \label{eq:PT.SF.eq50b} (a-N)(b-N)=0,
\end{equation} where $N\in\N$, $a=\ep_1\sqrt{2A}$ and
$b=\ep_2\sqrt{-2B}$, with $\ep_j=\pm1$, $j=1,2,3$,
independently.\end{theorem} 

\begin{proof}{See Okamoto \cite{refOkamotoPV}, Masuda \cite{refMasuda04} and Watanabe \cite{refWat}; also \cite[\S40]{refGLS02}.}\end{proof}

In the case when $N=0$ in (\ref{eq:PT.SF.eq50a}), then (\ref{eq:PT.DE.PVgen}) has solutions in terms of the associated Riccati equation
\begin{equation}\label{eq:PT.SF.eq51} z\deriv{w}{z}=aw^2+(b-a+\ep_3z)w-b.\end{equation} 
If $a\not=0$, then (\ref{eq:PT.SF.eq51}) has the solution 
\begin{equation*}\label{eq:PT.SF.eq52}
w(z)=-\frac{z}{a}\,\deriv{}{z} \ln\ph(z), 
\end{equation*} where $\ph(z)$ satisfies
\begin{equation*}\label{eq:PT.SF.eq53}
{z}^{2}\deriv[2]{\ph}{z}  + z (a -b+1-\ep_3 z ) \deriv{\ph}{z}  -ab\ph=0,
\end{equation*} 
which has solution
\begin{equation*}\label{eq:PT.SF.eq54}\fl
\ph(z)=\cases{z^b\left\{C_1 M(b, 1+a+b, z)+C_2U(b, 1+a+b, z)\right\},&$\mbox{\rm if}\enskip \ep_3=1,$\\ 
z^b\e^{-z}\left\{C_1 M(1+a, 1+a+b, z)+C_2U(1+a, 1+a+b, z)\right\},&$\mbox{\rm if}\enskip \ep_3=-1,$}
\end{equation*}  
with $C_1$ and $C_2$ arbitrary constants, and where $M(\a,\b, z)$ and $U(\a,\b, z)$ are Kummer functions. \hide{, or equivalently
\begin{equation*}
\fl\ph(z)=\cases{
{z^{-(a-b+1)/2}}\,\e^{z/2}\left\{{\widetilde{C}_1\WhitM{\k}{\mu}(z)+\widetilde{C}_2\WhitW{\k}{\mu}(z)}\right\},&$\mbox{\rm if}\enskip \ep_3=1$,\\
{z^{-(a-b+1)/2}}\,\e^{-z/2}\left\{{\widetilde{C}_1\WhitM{-\k}{\mu}(z)+\widetilde{C}_2\WhitW{-\k}{\mu}(z)}\right\},\quad&$\mbox{\rm if}\enskip \ep_3=-1$,}
\end{equation*} with $\widetilde{C}_1$ and $\widetilde{C}_2$ arbitrary constants, 
where $\WhitM{\k}{\mu}(z)$ and $\WhitW{\k}{\mu}(z)$ are Whittaker functions, 
$\k=\tfrac12(a-b+1)$ and $\mu=\tfrac12(a+b)$. }%

\subsection{\label{ssec:pvham}Hamiltonian structure for the fifth \peq.}
The Hamiltonian associated with (\ref{eq:PT.DE.PVgen}) is 
\begin{eqnarray}\fl{\HV}(q,p,z) 
= q(q-1)^2p^2  &-\{(b+\th) q^2-(2b+\th-z)q+b\}p 
-\tfrac14\{a^2-(b+\th)^2\}q,\label{ham5}\end{eqnarray} with $a$, $b$ and $\th$ parameters, 
where $p$ and $q$ satisfy
\begin{equation}\label{eq:PT.HM.DE523}\begin{array}{l}
\fl\ds z\deriv{q}{z}=2q(q-1)^2p-(b+\th) q^2+(2b+\th-z)q-b,\\[7.5pt] 
\fl\ds z\deriv{p}{z}=-(3q-1)(q-1)p^2+2(b+\th) qp-(2b+\th-z)p+\tfrac14\{a^2-(b+\th)^2\},
\end{array}\end{equation}
see Jimbo and Miwa \cite{refJM} and Okamoto \cite{refOkamoto80a,refOkamoto80b,refOkamotoPV}.
Eliminating $p$ then $q=w$ satisfies (\ref{eq:PT.DE.PVgen}) with parameters $(A,B,C)=(\tfrac12a^2,-\tfrac12b^2,-\th-1)$.

The second-order, second-degree equation satisfied by the Hamiltonian function is given in the following theorem.

\begin{theorem}{The Hamiltonian function
\begin{equation}\sV(z;a,b,\th)=\HV(q,p,z)+\tfrac14(2b+\th)z-\tfrac18(2b+\th)^2,\end{equation}  
with $\HV(q,p,z)$ given by (\ref{ham5}),
satisfies the second-order, second-degree equation 
\begin{equation}\left(z\deriv[2]{\sV}{z}\right)^{2} -\left\{2\left(\deriv{\sV}{z}\right)^{2}
-z\deriv{\sV}{z}+\sV\right\}^2
+4\prod_{j=1}^4\left(\deriv{\sV}{z}+\k_j\right)=0,\label{eq:SV}\end{equation}
for the parameters
\begin{equation}\fl\k_1=\tfrac14\th+\tfrac12a,\quad\k_2=\tfrac14\th-\tfrac12a,\quad\k_3=-\tfrac14\th-\tfrac12b,\quad\k_4=-\tfrac14\th+\tfrac12b.\label{SV:params}\end{equation}
Conversely, if $\sV(z;a,b,\th)$ satisfies (\ref{eq:SV}) then the solution of the Hamiltonian system (\ref{eq:PT.HM.DE523}) is given by
\begin{equation}\fl
q(z)  =\frac {\ds z\sV''+2(\sV')^{2}-z\sV' +\sV}{\ds 2(\sV'+\tfrac14\th-\tfrac12a) (\sV' +\tfrac14\th+\tfrac12a)},\quad
p(z)  ={\frac {\ds z\sV''-2(\sV')^{2}+z\sV'-\sV }{\ds 2(\sV'-\tfrac14\th+\tfrac12b)}},\label{eq:PT.HM5.DE35}
\end{equation}where $'=\d/\d z$.
}\end{theorem}
\begin{proof}{See Jimbo and Miwa \cite{refJM} and Okamoto \cite{refOkamoto80b,refOkamotoPV}.}\end{proof}

The special function solutions of (\ref{eq:SV}) are given in the following theorem.

\begin{theorem}{Let $\W_n(\psi)$ be the determinant given by
\begin{equation}\label{def:Wnn}
\W_n(\psi)=\left|\matrix{ \psi & \delta_{z}(\psi) & \cdots & \delta_{z}^{n-1}(\psi)\cr
\delta_{z}(\psi) & \delta_{z}^2(\psi) & \cdots & \delta_{z}^{n}(\psi)\cr
\vdots & \vdots & \ddots & \vdots \cr
\delta_{z}^{n-1}(\psi) & \delta_{z}^{n}(\psi) & \cdots & \delta_{z}^{2n-2}(\psi)
}\right|,\qquad \delta_{z}\equiv z\deriv{}{z}, \end{equation}
and define $\ph_{\a,\b}(z)$ by
\begin{equation*}\label{def:phi}
\ph_{\a,\b}(z)=C_1M(\a,\b,z)+C_2U(\a,\b,z),\end{equation*}
with $C_1$ and $C_2$ arbitrary constants, and where $M(\a,\b,z)$ and $U(\a,\b,z)$ are Kummer functions.
Then special function solutions of (\ref{eq:SV}) are given by 
\numparts\begin{eqnarray} \sigma(z) 
&=z\deriv{}{z}\big(\ln\W_n(\ph_{\a,\b})\big)-\tfrac14(3n+2\a-\b-1)z \nonumber \\ &\qquad 
-\tfrac58n^2-\tfrac14(2\a-3\b-1)n-\tfrac18(2\a-\b-1)^2,  \label{sol:SVi} \\ 
\sigma(z) 
&=z\deriv{}{z}\big(\ln\W_n(z^{\b}\ph_{\a,\b})\big)-\tfrac14(3n+2\a-\b-1)z \nonumber\\ &\qquad 
-\tfrac58n^2-\tfrac14(2\a+\b-1)n-\tfrac18(2\a-\b-1)^2,\label{sol:SVii}
\end{eqnarray}
for the parameters
\begin{equation}\begin{array}{l@{\qquad}l}
\k_1=\tfrac14(2\a-\b-n-1),&\k_3=\tfrac14(2\a-\b+3n-1),\\[5pt]
\k_2=-\tfrac14(2\a+\b+n-3),&\k_4=-\tfrac14(2\a-3\b+n+1),
\end{array}\end{equation}\endnumparts
and
\numparts\begin{eqnarray} \sigma(z) 
&=z\deriv{}{z}\big(\ln\W_n(z^{\b}\e^{-z}\ph_{\a,\b})\big)+\tfrac14(3n-2\a+\b-1)z \nonumber\\ &\qquad 
-\tfrac58n^2+\tfrac14(2\a-3\b+1)n-\tfrac18(2\a-\b+1)^2,\label{sol:SViii} \end{eqnarray}
for the parameters
\begin{equation}\begin{array}{l@{\qquad}l}
\k_1=\tfrac14(2\a-\b+n+1),&\k_3=\tfrac14(2\a-\b-3n+1),\\[5pt]
\k_2=-\tfrac14(2\a+\b-n-1),&\k_4=-\tfrac14(2\a-3\b-n+3).\end{array}
\end{equation}\endnumparts
}\end{theorem}
\begin{proof}{This result can be inferred from the work of Forrester and Witte \cite{refFW02} and Okamoto \cite{refOkamotoPV}.}\end{proof}


\section{\label{sec:op}Orthonormal polynomials}
\subsection{\label{ssec:cop}Continuous orthonormal polynomials}
Let $p_n(x)$, for $n\in\N$, be the orthonormal polynomial of degree $n$ in $x$ with respect to a positive weight $\w(x)$ 
on $(a,b)$, which a finite or infinite open interval in $\R$, such that
\begin{equation*}
\intS p_m(x)p_n(x)\,\w(x)\,\d x =\delta_{m,n},
\end{equation*}
with $\delta_{m,n}$ the Kronekar delta. 
One of the most important properties of orthogonal polynomials  is that they satisfy a three-term recurrence relationship of the form
\begin{equation}\label{eq:rr}
xp_n(x)=a_{n+1}p_{n+1}(x)+b_np_n(x)+a_np_{n-1}(x),
\end{equation}
where the coefficients $a_n$ and $b_n$ are given by the integrals
\begin{equation*}
a_n=\int_a^b  xp_n(x)p_{n-1}(x)\,\w(x)\,\d x,\qquad b_n=
\int_a^b xp_n^2(x)\,\w(x)\,\d x,\end{equation*}
with $p_{-1}(x)=0$. The coefficients in the recurrence relationship (\ref{eq:rr}) can also be expressed in terms of determinants whose coefficients are given in terms of the moments associated with the weight $\w(x)$. Specifically, the coefficients $a_n$ and $b_n$ in the recurrence relation (\ref{eq:rr}) are given by
\begin{equation}\label{eq:anbn}
a_n^2 = \frac{\Delta_{n+1}\Delta_{n-1}}{\Delta_{n}^2},\qquad b_n = \frac{\widetilde{\Delta}_{n+1}}{\Delta_{n+1}}-\frac{\widetilde{\Delta}_n}{\Delta_{n}},\end{equation}
where the determinants $\Delta_n$ and $\widetilde{\Delta}_n$ are given by
\begin{equation}\label{eq:dets} \fl 
\Delta_n=\left|\matrix{\mu_0 & \mu_1 & \ldots & \mu_{n-1}\cr
\mu_1 & \mu_2 & \ldots & \mu_{n}\cr
\vdots & \vdots & \ddots & \vdots \cr
\mu_{n-1} & \mu_{n} & \ldots & \mu_{2n-2}}\right|,\qquad
\widetilde{\Delta}_n=\left| \matrix{\mu_0 & \mu_1 & \ldots & \mu_{n-2} & \mu_n\cr
\mu_1 & \mu_2 & \ldots & \mu_{n-1}& \mu_{n+1}\cr
\vdots & \vdots & \ddots & \vdots & \vdots \cr
\mu_{n-1} & \mu_{n} & \ldots & \mu_{2n-3}& \mu_{2n-1} }\right|, \end{equation}
for $n=0,1,2,\ldots\ $,
with $\Delta_0=1$, $\Delta_{-1}=0$ and $\widetilde{\Delta}_0=0$, and
$\mu_k$, the $k$th moment, is given by the integral
\begin{equation}\label{eq:moment}
\mu_k=\intS  x^k\w(x)\,\d x
.\end{equation}

A characterization of \textit{classical} orthogonal polynomials (such as Hermite, Laguerre and Jacobi polynomials), is that their weights satisfy the \textit{Pearson equation}
\begin{equation}\label{eq:Pearson}
\deriv{}{x}[\sigma(x)\w(x)]=\tau(x)\w(x),
\end{equation}
where $\sigma(x)$ is a monic polynomial with deg$(\sigma)\leq2$ and $\tau(x)$ is a polynomial with deg$(\tau)=1$, 
cf.~\cite{refAlNod,refBochner,refChihara}. If the weight function $\w(x)$ satisfies the Pearson equation (\ref{eq:Pearson}) with either deg$(\sigma)>2$ or deg$(\tau)>1$, then the orthogonal polynomial is said to be \textit{semi-classical}, cf.~\cite{refHvR,refMaroni}. 

For further information about orthogonal polynomials see, for example, the books by 
Chihara \cite{refChihara}, Ismail \cite{refIsmail} and Szeg\"o \cite{refSzego}.

\subsection{\label{ssec:dop}Discrete orthonormal polynomials}
One can also define orthogonal polynomials on an equidistant lattice, rather than an interval. Consider the discrete orthonormal polynomials $\{p_n(x)\}$, $n=0,1,2,\ldots\ $, with respect to a discrete weight $\w(k)$ on the lattice $\N$
\begin{equation*}\sum_{k=0}^\infty p_n(k)p_m(k)\w(k)= \delta_{m,n},
\end{equation*}
which also satisfy the recurrence relation (\ref{eq:rr}).

The moments $\mu_n$ of the discrete weight $\w(k)$ are given by
\begin{equation*}\mu_n=\sum_{k=0}^\infty  k^n \w(k),\qquad n=0,1,2,\ldots\ ,\end{equation*}
and, as for the continuous orthonormal polynomials in \S\ref{ssec:cop} above, the coefficients in the recurrence relation are given by (\ref{eq:anbn}), with the determinants $\Delta_n$ and $\widetilde{\Delta}_n$ given by (\ref{eq:dets}).

In the special case when the discrete weight has the special form 
\begin{equation*}\w(k)=c(k)z^k,\qquad z>0,\end{equation*}
which is the case for the Charlier polynomials $C_n(k;z)$ and the Meixner polynomials $M_n(k;\a,z)$\ (see \S\ref{ssec:charlier} and \S\ref{ssec:meixner} below, respectively), then 
\begin{equation}\label{def:mun}
\mu_n(z)=\sum_{k=0}^\infty  k^n c(k){z^k} =\dz{z}^n(\mu_0),\qquad \dz{z}\equiv z\deriv{}{z}.\end{equation}
Consequently the determinants $\Delta_n(z)$ and $\widetilde{\Delta}_n(z)$ given by (\ref{eq:dets}) have the form
\begin{eqnarray*}\Delta_n(z)\hide{=\left|\matrix{\mu_0 & \mu_1 & \ldots & \mu_{n-1}\cr
\mu_1 & \mu_2 & \ldots & \mu_{n}\cr
\vdots & \vdots & \ddots & \vdots \cr
\mu_{n-1} & \mu_{n} & \ldots & \mu_{2n-2}}\right|} 
&=\left|\matrix{\mu_0 & \delta_{z}(\mu_0) & \ldots & \delta_{z}^{n-1}(\mu_{0})\cr
 \delta_{z}(\mu_0) & \delta_{z}^2(\mu_0) & \ldots & \delta_{z}^n(\mu_{0})\cr
\vdots & \vdots & \ddots & \vdots \cr
\delta_{z}^{n-1}(\mu_{0}) & \delta_{z}^{n}(\mu_{0}) & \ldots & \delta_{z}^{2n-2}(\mu_{0})}\right|,\\
\widetilde{\Delta}_n(z) &=\left|\matrix{ \mu_0 & \delta_{z}(\mu_0) & \ldots &  \delta_{z}^{n-2}(\mu_{0}) & \delta_{z}^n(\mu_{0})\cr
\delta_{z}(\mu_0) & \delta_{z}^2(\mu_0) & \ldots & \delta_{z}^{n-1}(\mu_{0})& \delta_{z}^{n+1}(\mu_{0})\cr
\vdots & \vdots & \ddots & \vdots & \vdots \cr
\delta_{z}^{n-1}(\mu_{0}) & \delta_{z}^{n}(\mu_{0}) & \ldots & \delta_{z}^{2n-3}(\mu_{0})& \delta_{z}^{2n-1}(\mu_{0})}\right|,\end{eqnarray*}
respectively. Hence we have the following result.

\begin{theorem}{\label{thm31}If the moment $\mu_n(z)$ has the form (\ref{def:mun}),
then the determinants $\Delta_n(z)$ and $\widetilde{\Delta}_n(z)$ can be written in the form
\begin{equation}\Delta_n(z)=\W_n(\mu_0),\qquad
\widetilde{\Delta}_n(z)=
\dz{z}\W_n(\mu_0),\label{DeltaDDelta}\end{equation}
where $\W_n(\psi)$ is defined by (\ref{def:Wnn}).}\end{theorem}

Some properties of the recurrence coefficients $a_n(z)$ and $b_n(z)$ are given in the following theorem.

\begin{theorem}{\label{thm32}If the moment $\mu_n(z)$ has the form (\ref{def:mun}),
then the recurrence coefficients $a_n(z)$ and $b_n(z)$ in (\ref{eq:rr}) satisfy the Toda system
\begin{equation}
\dz{z}\big(a_{n}^2\big)=a_{n}^2(b_{n}-b_{n-1}),\qquad \dz{z}\big(b_n\big)=a_{n+1}^2-a_{n}^2.\label{eq:toda}\end{equation}
Further the determinant $\Delta_n(z)$ satisfies the Toda equation
\begin{equation*}\dz{z}^2\big(\ln\Delta_n\big)=\frac{\Delta_{n+1}\Delta_{n-1}}{\Delta_n^2},\qquad n=1,2,\ldots\ .
\end{equation*}
}\end{theorem}

\begin{proof}{See \cite{refKMNOY,refOkamotoPIII,refSogo}; see also \cite{refNZ}. 
}\end{proof}

\begin{theorem}{\label{thm33}The recurrence coefficients $a_n(z)$ and $b_n(z)$ in (\ref{eq:rr}) can be expressed in the form
\begin{equation}a_n^2(z)=\dz{z}^2\big(\ln\W_n(\mu_0)\big),\qquad b_n(z) = 
\dz{z}\left(\ln\frac{\W_{n+1}(\mu_0)}{\W_n(\mu_0)}\right).\label{deq:anbn}\end{equation}}\end{theorem}

\begin{proof}{Applying Theorem \ref{thm31} to (\ref{eq:anbn}) gives the result.}\end{proof}

As discussed in \S\ref{ssec:cop} above, classical orthogonal polynomials are characterized by the Pearson equation (\ref{eq:Pearson}). Analogously discrete orthogonal polynomials are characterized by the discrete Pearson equation
\begin{equation}\Delta\big[\sigma(k)\w(k)\big]=\tau(k)\w(k),\label{eq:disPearson}\end{equation}
where $\Delta$ is the forward difference operator
\[\Delta f(k)=f(k+1)-f(k),\]
$\sigma(k)$ is a monic polynomial with deg$(\sigma)\leq2$ and $\tau(k)$ is a polynomial with deg$(\tau)=1$.
A discrete Pearson equation can also be defined using the backward difference operator
\[\nabla f(k)=f(k)-f(k-1).\]
If the discrete weight $\w(k)$ satisfies (\ref{eq:disPearson}) with either deg$(\sigma)>2$ or deg$(\tau)>1$, then the discrete orthogonal polynomial is said to be \textit{semi-classical}, cf.~\cite{refDM}. 

For further information about discrete orthogonal polynomials see, for example, the books by Beals and Wong \cite[Chapter 5]{refBW},
Chihara \cite[Chapter VI]{refChihara}, Ismail \cite[Chapter 6]{refIsmail} and Nikiforov, Suslov and Uvarov \cite{refNSU}.

\section{\label{sec:char}Charlier polynomials and generalized Charlier polynomials}
\subsection{\label{ssec:charlier}Charlier polynomials}The {Charlier polynomials} $C_n(k;z)$ 
are a family of orthogonal polynomials introduced in 1905 by Charlier \cite{refChar}
given by
\begin{equation}
C_n(k;z) =\HyperpFq{2}{0}\left(-n,-k;;-\ifrac{1}{z}\right)=(-1)^nn!L_n^{(-1-k)}\left(-\ifrac{1}{z}\right),\quad z>0,
\label{charlier}\end{equation}
where $\HyperpFq{2}{0}(a,b;;z)$ is the hypergeometric function 
and $L_n^{(\a)}(z)$ is the {associated Laguerre polynomial}, see \cite{refBW,refChihara,refIsmail,refDLMF}.
The Charlier polynomials are orthogonal on the lattice $\N$ with respect to the Poisson distribution 
\begin{equation} \w(k) =\frac{z^k}{k!},\qquad  z>0,\label{charlierw}\end{equation}
and satisfy the orthogonality condition
\[\sum_{k=0}^\infty C_m(k;z)C_n(k;z)\frac{z^k}{k!}=\frac{n!\,\e^{z}}{z^n}\delta_{m,n}.\]
The weight (\ref{charlierw}) satisfies the discrete Pearson equation (\ref{eq:disPearson}) with 
\[\sigma(k)=k,\qquad\tau(k)=z-k.\]
From (\ref{charlierw}), the moment $\mu_0(z)$ is given by
\begin{equation*}\mu_0(z) 
= \sum_{k=0}^\infty \frac{z^k}{k!} = \e^{z}.\end{equation*}
Hence from Theorem \ref{thm31}, the Hankel determinant $\Delta_n(z)$ is given by
\begin{equation*}
\Delta_n(z)=\W_n(\mu_0)=z^{n(n-1)/2}\e^{nz}\,\prod_{k=1}^{n-1}(k!),
\end{equation*}
and so from Theorem \ref{thm33} the recurrence coefficients are given by
\begin{eqnarray*} a_n^2(z)=\dz{z}^2\big(\ln\W_n(\mu_0)\big)=nz,\qquad
b_n(z) =\dz{z}\left(\ln\frac{\W_{n+1}(\mu_0)}{\W_n(\mu_0)}\right)= n+z. 
\end{eqnarray*}

\subsection{\label{ssec:gencharlier}Generalized Charlier polynomials }
Smet and van Assche \cite{refSmetvA} generalized the Charlier weight (\ref{charlierw}) with one additional parameter through the weight function
\[\w(x)=\frac{\Gamma(\b+1)\,z^x}{\Gamma(\b+x+1)\,\Gamma(x+1)},\qquad  z>0,\] 
with  $\b$ a parameter such that $\b>-1$. This gives the discrete weight 
\begin{equation} \w(k) = \frac{z^k}{(\b+1)_k\,k!},\qquad  z>0,\label{gencharlierw}\end{equation}
where $(\b+1)_k
=\Gamma(\b+1+k)/\Gamma(\b+1)$ is the Pochhammer symbol, on the lattice $\N$.
The weight (\ref{gencharlierw}) satisfies the discrete Pearson equation (\ref{eq:disPearson}) with 
\[\sigma(k)=k(k+\b),\qquad\tau(k)=-k^2-\b k+z,\]
and so the generalized Charlier polynomials are semi-classical orthogonal polynomials.
The special case $\b=0$ was first considered by Hounkonnou, Hounga and Ronveaux \cite{refHHR} and later studied by van Assche and Foupouagnigni \cite{refvAF03}.  

For the generalized Charlier weight (\ref{gencharlierw}), the orthonormal polynomials $p_n(k;z)$ satisfy the orthogonality condition
\[\sum_{k=0}^\infty p_m(k;z)p_n(k;z)\frac{z^k}{(\b+1)_k\,k!}=\delta_{m,n},\]
and the three-term recurrence relation
\begin{equation}
xp_n(x;z)=a_{n+1}(z)p_{n+1}(x;z)+b_n(z)p_n(x;z)+a_n(z)p_{n-1}(x;z),
\label{deq:rr1}\end{equation}
with $p_1(x;z)=0$ and $p_0(x;z)=1$. Our interest is determining explicit expressions for the coefficients 
$a_n(z)$ and $b_n(z)$ in the recurrence relation (\ref{deq:rr1}).

Smet and van Assche  \cite[Theorem 2.1]{refSmetvA} proved the following theorem for recurrence coefficients associated with the generalized Charlier weight (\ref{gencharlierw}).

\begin{theorem}{The recurrence coefficients $a_n(z)$ and $b_n(z)$ for orthonormal polynomials associated with the generalized Charlier weight (\ref{gencharlierw}) on the lattice $\N$ satisfy the discrete system
\begin{equation}\label{SvAsys}\begin{array}{l}
(a_{n+1}^2-z)(a_{n}^2-z)=z(b_n-n)(b_n-n+\b),\\[5pt] \ds b_n+b_{n-1}-n+\b+1=\ifrac{nz}{a_n^2},
\end{array}\end{equation}
with initial conditions
\begin{equation} \label{SvAsysic}
a_0^2=0,\qquad b_0=\frac{\sqrt{z}\,\BesselI{\b+1}(2\sqrt{z})}{\BesselI{\b}(2\sqrt{z})}=z\deriv{}{z}\ln\big(z^{-\b/2}\BesselI{\b}(2\sqrt{z})\big),
\end{equation}
with $\BesselI{\nu}(x)$ the {modified Bessel function}.}\end{theorem}

Smet and van Assche \cite[Theorem 2.1]{refSmetvA} show that the system (\ref{SvAsys}) is a limiting case of a discrete \p\ equation, namely the first \dPIV\ in \cite[p.\ 723]{refvA07}. Since the initial conditions (\ref{SvAsysic}) involve modified Bessel functions, then clearly the solutions of the discrete system (\ref{SvAsys}) are expressed in terms modified Bessel functions.

Using the discrete system (\ref{SvAsys}) and the Toda system (\ref{eq:toda}),
Filipuk and van Assche \cite{refFvA13} show that the recurrence coefficient $b_n$ can be expressed in terms of solutions of a special case of \PV\ which can be transformed into \PIII. However their proof is rather involved and several details are omitted due to the size of expressions involved. Further Filipuk and van Assche \cite{refFvA13} do not obtain explicit expressions for the the recurrence coefficients $a_n$ and $b_n$.

The relationship between the recurrence coefficients $a_n(z)$ and $b_n(z)$ and classical solutions of \PIIIi\ can be shown in much more straightforward way and also obtain explicit expressions for these coefficients. 
First we obtain explicit expressions for the moment $\mu_0(z)$ and the Hankel determinant $\Delta_n(z)$.

\begin{theorem}{For the generalized Charlier weight (\ref{gencharlierw}) the moment $\mu_0(z)$ is given by
\begin{equation}
\mu_0(z)=\sum_{k=0}^\infty \frac{z^k}{(\b+1)_k\,k!}
=\Gamma(\b+1)z^{-\b/2}\BesselI{\b}\big(2\sqrt{z}\big),\label{genC:mu0}\end{equation}
with $\BesselI{\nu}(x)$ the {modified Bessel function},
and the Hankel determinant $\Delta_n(z;\b)$ is given by
\begin{equation}\Delta_n(z;\b)=\big[\Gamma(\b+1)\big]^n \W_n\Big(z^{-\b/2}\BesselI{\b}\big(2\sqrt{z}\big)\Big).\label{genC:Delta}
\end{equation}}\end{theorem}

\begin{proof}Since the modified Bessel function $\BesselI{\nu}(x)$ has the series expansion \cite[\S10.25.2]{refDLMF}
\[\BesselI{\nu}(x)=
\sum _{{k=0}}^{\infty}\frac{(\tfrac12x)^{2k+\nu}}{k!\mathop{\Gamma\/}\nolimits\!\left(\nu+k+1\right)},\]
then the expression (\ref{genC:mu0}) for the moment $\mu_0(z)$ follows immediately. Then using Theorem \ref{thm31}, we obtain the expression (\ref{genC:Delta}) for the Hankel determinant $\Delta_n(z;\b)$.
\end{proof}

Hence we obtain explicit expressions for the recurrence coefficients $a_n(z)$ and $b_n(z)$.

\begin{corollary}{The coefficients $a_n(z)$ and $b_n(z)$ in the recurrence relation (\ref{deq:rr1}) have the form
\begin{equation}a_n^2(z)=\left(z\deriv{}{z}\right)^2\big(\ln\Delta_n(z;\b)\big),\qquad b_n(z) = 
z\deriv{}{z}\left(\ln\frac{\Delta_{n+1}(z;\b)}{\Delta_n(z;\b)}\right),\label{deq:anbn1}\end{equation}
with $\Delta_n(z;\b)$ given by (\ref{genC:Delta}).}\end{corollary}
\begin{proof}{These follow immediately from Theorem \ref{thm33}.}\end{proof}

Finally we relate the Hankel determinant $\Delta_n(z;\b)$ to solutions of (\ref{eq:PT.HM3p.DE34}), the \PIIIi\ $\sigma$-equation.

\begin{theorem}{The function 
\begin{equation} S_n(z;\b)=z\deriv{}{z}\ln\Delta_n(z;\b), \label{def:Sn3}\end{equation} 
with $\Delta_n(z;\b)$ given by (\ref{genC:Delta}),
satisfies the second-oder, second-degree equation
\begin{eqnarray}\left(z\deriv[2]{S_n}{z} \right)^{\!2}&=\left[ n- \left( n+\b \right) \deriv{S_n}{z}  \right] ^2 \nonumber\\ &\quad
-4 \deriv{S_n}{z}  \left(\deriv{S_n}{z} -1 \right)\left[z\deriv{S_n}{z} -S_n + \tfrac12n (n-1)\right].\label{eq:Sp3i}\end{eqnarray}}\end{theorem}

\begin{proof}{\Eref{eq:Sp3i} is equivalent to 
(\ref{eq:PT.HM3p.DE34}) through the transformation
\begin{equation}S_n(z;\b) =\sIII (z) +\tfrac12{z}+\tfrac14{n}^2-\tfrac12n(\b+1)-\tfrac14\b^{2},\label{eq:Sn3}\end{equation}
with parameters $(\th_0,\th_{\infty})=(n+\b,n-\b)$, as is easily verified. Then comparing (\ref{eq:Sn3}), with $S_n$ given by (\ref{def:Sn3}), to (\ref{sol:SIIIi}), with $\nu=\b$, $\ep_1=-1$ and $\ep_2=1$, gives the result.
}\end{proof}

\begin{remarks}{\quad \phantom{x}
\begin{enumerate}
\item In terms of $S_n(z;\b)$ given by (\ref{def:Sn3}), the coefficients $a_n(z)$ and $b_n(z)$ in the recurrence relation (\ref{deq:rr1}) have the form 
\[a_n^2(z)=z\deriv{}{z} S_n(z;\b),\qquad b_n(z) =S_{n+1}(z;\b)-S_n(z;\b).\]
\item Solving (\ref{eq:Sn3}) for $\sIII (z)$ and substituting in (\ref{eq:PT.HM3p.DE35}) yields the solution of the Hamiltonian system (\ref{eq:PT.HM3p.DE323}) given by
\[q(z)=\frac{zS_n''(z)-2(n+\b)S_n'(z)+2n}{4S_n'(z)\big[1-S_n'(z)\big]},\qquad p(z)=S_n'(z),\qquad z=t,\] with parameters $(\th_0,\th_{\infty})=(n+\b,n-\b)$.
\end{enumerate}}\end{remarks}

\section{\label{sec:meix}Meixner polynomials and generalizations}
\subsection{\label{ssec:meixner}Meixner polynomials}The {Meixner polynomials} $M_n(k;\a,z)$  
are a family of discrete orthogonal polynomials introduced in 1934 by Meixner \cite{refMeix}
given by
\begin{equation} M_n(k;\a,z)=\HyperpFq{2}{1}\left(-n,-k;-\a;1-\ifrac{1}{z}\right),\qquad0<z<1,\label{meixner}\end{equation}
with $\a>0$,
where $\HyperpFq{2}{1}(a,b;c;z)$ is the hypergeometric function, see \cite{refBW,refChihara,refIsmail,refDLMF}. In the case when $\a=-N$, with $N\in\N$ and $z=p/(1-p)$, these polynomials are referred to as the Krawtchouk polynomials for $k\in\{0,1,\ldots,N\}$
\begin{equation}
K(k;p,N)=\HyperpFq{2}{1}\left(-n,-k;N;\ifrac{1}{p}\right),\qquad 0<p<1,
\label{krawtchouk}\end{equation}
which were introduced in 1929 by Krawtchouk \cite{refKraw}.
The Meixner polynomials (\ref{meixner}) are orthogonal with respect to the discrete weight
\begin{equation} \w(k) =\frac{(\a)_k\,z^k}{k!},\qquad \a>0,\quad z>0,\label{meixnerw}\end{equation}
and satisfy the orthogonality condition
\[\sum_{k=0}^\infty M_m(k;\a,z)M_n(k;\a,z)\frac{(\a)_k\,z^k}{k!}=\frac{n!\,z^{-n}}{(\a)_n (1-z)^\a}\delta_{m,n}.\]
The weight (\ref{meixnerw}) satisfies the discrete Pearson equation (\ref{eq:disPearson}) with 
\[\sigma(k)=k,\qquad\tau(k)=(z-1)k +z\a.\]
From (\ref{meixnerw}), the moment $\mu_0(z)$ is given by
\begin{equation*}\mu_0(z) 
=\sum_{k=0}^\infty \frac{(\a)_k\,z^k}{k!} =(1-z)^{-\a}.\end{equation*}
Hence from Theorem \ref{thm31} that the Hankel determinant $\Delta_n(z)$ is given by
\begin{equation*}
\Delta_n(z)=\W_n(\mu_0)=\frac{z^{n(n-1)/2}}{(1-z)^{n(n+\a-1)}}\,\prod_{k=1}^{n-1}k!(\a+k)^{n-k-1},
\end{equation*}
and so, from Theorem \ref{thm33}, the recurrence coefficients are given by
\begin{eqnarray*}a_n^2(z)=\frac{n(n+\a-1)z}{(1-z)^2},\qquad b_n(z) = \frac{n+(n+\a)z}{1-z}.
\end{eqnarray*}

\subsection{\label{ssec:genmeixner}Generalized Meixner polynomials}
In a similar way to that for the Charlier weight above, Smet and van Assche \cite{refSmetvA} generalized the Meixner weight (\ref{charlierw}) with one additional parameter through the weight function the weight function
\[\w(x)=\frac{\Gamma(\a+x)\,\Gamma(\b)\,z^x}{\Gamma(\a)\,\Gamma(\b+x)\,\Gamma(x+1)},\qquad z>0,\] 
with $\a,\b>0$, which gives the weight 
\begin{equation}\w(k) =\frac{(\a)_k\,z^k}{(\b)_k\,k!},\qquad z>0.\label{genmeixnerw}\end{equation}
The weight (\ref{genmeixnerw}) satisfies the discrete Pearson equation (\ref{eq:disPearson}) with 
\[\sigma(k)=k(k+\b-1),\qquad\tau(k)=-k^2+(z+1-\b)k+z\a,\]
and so the generalized Meixner polynomials are semi-classical orthogonal polynomials.

Boelen, Filipuk and van Assche \cite{refBFvA} considered the special case of (\ref{genmeixnerw}) when $\b=1$ and showed that
the recurrence coefficients $a_n$ and $b_n$ satisfy a limiting case of an asymmetric \dPIV\ equation. We note that the special case $\a=\b$ gives the classical Charlier weight (\ref{charlierw}) and the case $\a=1$ corresponds to the classical Charlier weight on the lattice $\N+1-\b$.

For the generalized Meixner weight (\ref{genmeixnerw}), the orthonormal polynomials $p_n(k;z)$ satisfy the orthogonality condition
\[\sum_{k=0}^\infty p_m(k;z)p_n(k;z)\frac{(\a)_k\,z^k}{(\b)_k\,k!}=\delta_{m,n},\]
and the three-term recurrence relation
\begin{equation}
xp_n(x;z)=a_{n+1}(z)p_{n+1}(x;z)+b_n(z)p_n(x;z)+a_n(z)p_{n-1}(x;z),
\label{deq:rr2}\end{equation}
with $p_1(x;z)=0$ and $p_0(x;z)=1$. As for the generalized Charlier weight (\ref{gencharlierw}), our interest is determining explicit expressions for the coefficients $a_n(z)$ and $b_n(z)$ in the recurrence relation (\ref{deq:rr2}).

Smet and van Assche  \cite[Theorem 2.1]{refSmetvA} proved the following theorem for recurrence coefficients associated with the generalized Meixner weight (\ref{genmeixnerw}).
\begin{theorem}{The recurrence coefficients $a_n(z)$ and $b_n(z)$ for orthonormal polynomials associated with the generalized Meixner weight (\ref{genmeixnerw}) on the lattice $\N$ satisfy
\begin{eqnarray*}a_n^2=nz-(\a-1)x_n,\\ b_n=n+\a-\b+z-(\a-1)y_n/z.\end{eqnarray*}
where $x_n$ and $y_n$ satisfy the discrete system
\begin{equation}\label{SvAsysM}\begin{array}{l}
\ds(x_n+y_n)(x_{n+1}+y_n)=\ds\frac{\a-1}{z^2}y_n(y_n-z)\left(y_n-z\,\frac{\a-\b}{\a-1}\right),\\[5pt]
\ds(x_n+y_n)(x_{n}+y_{n-1})=\ds\frac{(\a-1)x_n(x_n+z)}{(\a-1)x_n-nz}\left(x_n+z\,\frac{\a-\b}{\a-1}\right), 
\end{array}\end{equation} with initial conditions
\begin{equation}\label{SvAsysMic}
a_0^2=0,\qquad b_0= \frac{\a z}{\b}\,\frac{M(\a+1,\b+1,z)}{M(\a,\b,z)}=z\deriv{}{z}\ln M(\a,\b,z),
\end{equation}
and $M(\a,\b,z)$ is the Kummer function. 
}\end{theorem}
We note that $M(\a,\b,z)=\HyperpFq{1}{1}(\a,\b,z)$, the confluent hypergeometric function \cite[\S13]{refDLMF}.
Smet and van Assche \cite{refSmetvA} show that the discrete system (\ref{SvAsysM}) can be identified as a limiting case of an asymmetric \dPIV\ equation. Filipuk and van Assche \cite{refFvA11} show that the system (\ref{SvAsysM}) can be obtained from a \bt\ of \PV\ (\ref{eq:PT.DE.PV}).
Since the initial conditions (\ref{SvAsysMic}) involve Kummer functions, then clearly the solutions of the discrete system (\ref{SvAsysM}) are expressed in terms Kummer functions.

Using the discrete system (\ref{SvAsysM}) and the Toda system (\ref{eq:toda}),
Filipuk and van Assche \cite{refFvA11} show that the recurrence coefficients $a_n(z)$ and $b_n(z)$ are related to classical solutions of \PV\ (\ref{eq:PT.DE.PV}), for the parameters
\[(A,B,C,D)=\big(\tfrac12(\a-1)^2,-\tfrac12(n+\a-\b)^2,\pm(n+\b),-\tfrac12\big).\]
However 
their proof is rather cumbersome and most of the details are omitted due to the size of expressions involved. Further Filipuk and van Assche \cite{refFvA11} do not obtain explicit expressions for the the recurrence coefficients $a_n$ and $b_n$.

In an analogous way to that for the generalized Charlier polynomials in \S\ref{ssec:gencharlier} above, the relationship between the recurrence coefficients $a_n(z)$ and $b_n(z)$ and classical solutions of \PV\ can be shown using a much more straightforward way and also obtain explicit expressions for these coefficients.
First we obtain explicit expressions for the moment $\mu_0(z)$ and the Hankel determinant $\Delta_n(z;\a,\b)$.

\begin{theorem}{For the generalized Meixner weight (\ref{genmeixnerw}) the moment $\mu_0(z)$ is given by
\begin{equation}
\mu_0(z) 
=\sum_{k=0}^\infty \frac{(\a)_k\,z^k}{(\b)_k\,k!}=M(\a,\b,z),\label{genM:mu0}\end{equation}
with $M(\a,\b,z)$ the \textit{Kummer function}, and the Hankel determinant $\Delta_n(z)$ given by
\begin{equation}\Delta_n(z;\a,\b)=\W_n\big(M(\a,\b,z)\big).\label{genM:Delta}
\end{equation}}\end{theorem}
\begin{proof}Since the Kummer function $M(\a,\b,z)$ has the series expansion \cite[\S13.2.2]{refDLMF}
\[M(\a,\b,z)=\sum _{k=0}^{\infty}\frac{(\a)_{k}}{(\b)_{k}\,k!}\,z^{{k}},\]
then the expression (\ref{genM:mu0}) for the moment $\mu_0(z)$ follows immediately. Then we use Theorem \ref{thm31}, to obtain the expression (\ref{genM:Delta}) for the Hankel determinant $\Delta_n(z;\a,\b)$.
\end{proof}

Hence we obtain explicit expressions for the recurrence coefficients $a_n(z)$ and $b_n(z)$.

\begin{theorem}{The coefficients $a_n(z)$ and $b_n(z)$ in the recurrence relation (\ref{deq:rr2}) have the form
\begin{equation}\fl
a_n^2(z)=\left(z\deriv{}{z}\right)^2\big(\ln\Delta_n(z;\a,\b)\big),\qquad b_n(z) = 
z\deriv{}{z}\left(\ln\frac{\Delta_{n+1}(z;\a,\b)}{\Delta_n(z;\a,\b)}\right),\label{deq:anbn2}\end{equation}
with $\Delta_n(z)$ given by (\ref{genM:Delta}).}\end{theorem}
\begin{proof}{These follow immediately from Theorem \ref{thm33}.}\end{proof}

Finally we relate the Hankel determinant $\Delta_n(z)$ to solutions of (\ref{eq:SV}), the \PV\ $\sigma$-equation.

\begin{theorem}{The function 
 \begin{equation} \ds S_n(z;\a,\b)=z\deriv{}{z} \ln\Delta_n(z;\a,\b),\label{def:Sn5}\end{equation}
with $\Delta_n(z;\a,\b)$ given by (\ref{genM:Delta}), satisfies the second-oder, second-degree equation
\begin{eqnarray} \fl \left(z\deriv[2]{S_n}{z} \right)^{\!2}=\left[(z+ n+\beta -1) \deriv{S_n}{z} -S_n-\tfrac12n(n-1+2\a) \right] ^2 \nonumber\\ 
-4 \deriv{S_n}{z}  \left(\deriv{S_n}{z}-n-\a+\b\right)\left[z\deriv{S_n}{z} -S_n+\tfrac12n(n-1) \right].\label{eq:sH5} \end{eqnarray}}\end{theorem}

\begin{proof}{\Eref{eq:sH5} is equivalent to 
(\ref{eq:SV}) through the transformation
\begin{eqnarray}S_n(z;\a,\b) =\sV (z) &+ \tfrac14(2\a-\b+3n-1)z\nonumber\\
&+\tfrac58{n}^{2}+ \tfrac14(2\a-3\b-1) n+\tfrac18 (2\a -\b-1) ^{2},
\label{eq:Sn5}\end{eqnarray}
with parameters
 \begin{equation}\label{Sn5:params}\begin{array}{l@{\qquad}l}
\k_1=\tfrac14(2\a-\b-n-1),&\k_3=\tfrac14(2\a-\b+3n-1),\\[5pt]
\k_2=-\tfrac14(2\a+\b+n-3),&\k_4=-\tfrac14(2\a-3\b+n+1),
\end{array}\end{equation}as is easily verified. Then comparing (\ref{eq:Sn5}), with $S_n$ given by (\ref{def:Sn5}), to (\ref{sol:SVi}) gives the result.
}\end{proof}

\begin{remarks}{\quad \phantom{x}
\begin{enumerate}
\item In terms of $S_n(z;\a,\b)$ given by (\ref{def:Sn5}), the coefficients $a_n(z)$ and $b_n(z)$ in the recurrence relation (\ref{deq:rr2}) have the form 
\[a_n^2(z)=z\deriv{}{z} S_n(z;\a,\b),\qquad b_n(z) =S_{n+1}(z;\a,\b)-S_n(z;\a,\b).\]
\item Substituting the parameters (\ref{Sn5:params}) in (\ref{SV:params}) yields
\begin{equation}\label{Ham5:params}(a,b,\th)=(\a-1,\b-\a-n,1-\b-n),\end{equation}
and so (\ref{eq:sH5}) is equivalent to  \PV\ for the  parameters
\[(A,B,C,D)=\big(\tfrac12(\a-1)^2,-\tfrac12(\a-\b+n),\b+n-2,-\tfrac12\big).\]
\item Solving (\ref{eq:Sn5}) for $\sV (z)$ and substituting in (\ref{eq:PT.HM5.DE35}) yields the solution of the Hamiltonian system 
(\ref{eq:PT.HM.DE523}) given by
\begin{eqnarray*}\fl
q(z)={\frac {2zS_n''+4[S_n']^{2} -2\left(z-2\alpha+\beta-3n+1 \right) S_n'+2S_n  +n(n+2\alpha-1)}{ 4\big[S_n' -\alpha-n+1 \big] \big[S_n' -n \big]}},\\
\fl p(z)={\frac {2zS_n''-4[S_n']^{2} +2\left(z-2\alpha+\beta-3n+1 \right) S_n'-2S_n  -n(n+2\alpha-1)}{ 4\big[S_n' -\alpha+\beta-n\big] }},\end{eqnarray*}
for the parameters (\ref{Ham5:params}).
\end{enumerate}}\end{remarks}

 \section{\label{sec:dis}Discussion}\def\la{\alpha}
In this paper we have studied semi-classical generalizations of the Charlier polynomials and the Meixner polynomials. These discrete orthogonal polynomials satisfy three-term recurrence relations whose coefficients depend on a parameter. We have shown that the coefficients in these recurrence relations can be explicitly expressed in terms of Wronskians of modified Bessel functions and Kummer functions, respectively. These Wronskians also arise in the description of special function solutions of the third and fifth \p\ equations and the second-order, second-degree equations satisfied by the associated Hamiltonian functions. The results in this paper are more comprehensive than those in \cite{refFvA13} for generalized Charlier polynomials and in \cite{refBFvA,refFvA11} for generalized Meixner polynomials. The link between the semi-classical discrete orthogonal polynomials and the special function solutions of the \p\ equations is the moment for the associated weight which enables the Hankel determinant to be written as a Wronskian. In our opinion, this  illustrates the increasing significance of the \p\ equations in the field of orthogonal polynomials and special functions.

\hide{Recurrence relations for orthogonal polynomials with respect to discontinuous weights which involve the Heaviside function $\mathcal{H}(x)$ have also been expressed in terms of solutions of \p\ equations.
For example, in terms of solutions of \PIV\ for the discontinuous Hermite weight \cite{refFW03}
\begin{equation*}\w(x;t)=\left\{1-\xi\mathcal{H}(x-t)\right\}|x-t|^{\la-1}\exp(-x^2),\qquad x\in\R,\end{equation*}
with $t\in\R$, $\la>0$ and $0<\xi<1$;  
in terms of solutions of \PV\  for the deformed Laguerre weight \cite{refBC09,refFO10}
\begin{equation*}\w(x;t)=\left\{1-\xi\mathcal{H}(x-t)\right\}|x-t|^\a x^\b\e^{-x}, \qquad x\in\R^+,\end{equation*}
with $t\in\R^+$, $\a,\b>0$ and $0<\xi<1$; and  
in terms of solutions of \PVI\ for the discontinuous Jacobi weight \cite{refCZ10}
\begin{equation*}\w(x;t)=\left\{1-\xi\mathcal{H}(x-t)\right\}x^\a(1-x)^\b,\qquad x\in[0,1],\end{equation*}
with $t\in[0,1]$, $\a,\b>0$ and $0<\xi<1$.}%

\ack I thank Kerstin Jordaan, Ana Loureiro, James Smith and Walter van Assche for their helpful comments and illuminating discussions.

\def\ams{American Mathematical Society}
\def\bM{Acta Appl. Math.}
\def\ARMA{Arch. Rat. Mech. Anal.}
\def\bull{Acad. Roy. Belg. Bull. Cl. Sc. (5)}
\def\AC{Acta Crystrallogr.}
\def\AM{Acta Metall.}
\def\ampa{Ann. Mat. Pura Appl. (IV)}
\def\AP{Ann. Phys., Lpz.}
\def\APNY{Ann. Phys., NY}
\def\APP{Ann. Phys., Paris}
\def\BAMS{Bull. Amer. Math. Soc.}
\def\CJP{Can. J. Phys.}
\def\cmp{Commun. Math. Phys.}
\def\CMP{Commun. Math. Phys.}
\def\cpam{Commun. Pure Appl. Math.}
\def\CPAM{Commun. Pure Appl. Math.}
\def\CQG{Classical Quantum Grav.}
\def\crp{C.R. Acad. Sc. Paris}
\def\CSF{Chaos, Solitons \&\ Fractals}
\def\DE{Diff. Eqns.}
\def\DU{Diff. Urav.}
\def\ejam{Europ. J. Appl. Math.}
\def\EJAM{Europ. J. Appl. Math.}
\def\funk{Funkcial. Ekvac.}
\def\FUNK{Funkcial. Ekvac.}
\def\IP{Inverse Problems}
\def\JAMS{J. Amer. Math. Soc.}
\def\JAP{J. Appl. Phys.}
\def\JCP{J. Chem. Phys.}
\def\JDE{J. Diff. Eqns.}
\def\JFM{J. Fluid Mech.}
\def\JJAP{Japan J. Appl. Phys.}
\def\JP{J. Physique}
\def\JPhCh{J. Phys. Chem.}
\def\JMAA{J. Math. Anal. Appl.}
\def\JMMM{J. Magn. Magn. Mater.}
\def\JMP{J. Math. Phys.}
\def\jmp{J. Math. Phys}
\def\JNMP{J. Nonl. Math. Phys.}
\def\JPA{J. Phys. A: Math. Theor.}
\def\JPB{J. Phys. B: At. Mol. Phys.} 
\def\jpb{J. Phys. B: At. Mol. Opt. Phys.} 
\def\JPC{J. Phys. C: Solid State Phys.} 
\def\JPCM{J. Phys: Condensed Matter} 
\def\JPD{J. Phys. D: Appl. Phys.}
\def\JPE{J. Phys. E: Sci. Instrum.}
\def\JPF{J. Phys. F: Metal Phys.}
\def\JPG{J. Phys. G: Nucl. Phys.} 
\def\jpg{J. Phys. G: Nucl. Part. Phys.} 
\def\JSP{J. Stat. Phys.}
\def\JOSA{J. Opt. Soc. Am.}
\def\JPSJ{J. Phys. Soc. Japan}
\def\JQSRT{J. Quant. Spectrosc. Radiat. Transfer}
\def\LMP{Lett. Math. Phys.}
\def\LNC{Lett. Nuovo Cim.}
\def\NC{Nuovo Cim.}
\def\NIM{Nucl. Instrum. Methods}
\def\NL{Nonlinearity}
\def\NMJ{Nagoya Math. J.}
\def\NP{Nucl. Phys.}
\def\pl{Phys. Lett.}
\def\PL{Phys. Lett.}
\def\PMB{Phys. Med. Biol.}
\def\PR{Phys. Rev.}
\def\PRL{Phys. Rev. Lett.}
\def\PRS{Proc. R. Soc.}
\def\prsl{Proc. R. Soc. Lond. A}
\def\PRSL{Proc. R. Soc. Lond. A}
\def\PS{Phys. Scr.}
\def\PSS{Phys. Status Solidi}
\def\PTRS{Phil. Trans. R. Soc.}
\def\RMP{Rev. Mod. Phys.}
\def\RPP{Rep. Prog. Phys.}
\def\RSI{Rev. Sci. Instrum.}
\def\SAM{Stud. Appl. Math.}
\def\sam{Stud. Appl. Math.}
\def\SSC{Solid State Commun.}
\def\SST{Semicond. Sci. Technol.}
\def\SUST{Supercond. Sci. Technol.}
\def\ZP{Z. Phys.}
\def\JCAM{J. Comput. Appl. Math.}

\def\OUP{Oxford University Press}
\def\CUP{Cambridge University Press}

\def\refpp#1#2#3#4#5{\bibitem{#1} \textrm{\frenchspacing#2}\  #5\ \textrm{#3} #4}

\def\refjl#1#2#3#4#5#6#7{\bibitem{#1} \textrm{\frenchspacing#2}\ #7\ \textrm{#6}
{\frenchspacing\it#3}\ \textbf{#4}\ #5}

\def\refjltoap#1#2#3#4#5#6#7{\bibitem{#1} \textrm{\frenchspacing#2}\ #7\ \textrm{#6} 
\textit{\frenchspacing#3}\ #5} 

\def\refbk#1#2#3#4#5{\bibitem{#1} \textrm{\frenchspacing#2}\ #5\ \textit{#3}\ #4} 

\def\refbkk#1#2#3#4#5{\bibitem{#1} \textrm{\frenchspacing#2}\ #5\ \textit{#3}\ (#4)} 

\def\refcf#1#2#3#4#5#6#7{\bibitem{#1} \textrm{\frenchspacing#2}\ #7\ {#3}
 \textit{#4}\ {\frenchspacing#5}\ #6}

\section*{References}

\end{document}

\newpage\section{Krawtchouk polynomials and generalized Krawtchouk polynomials}
\subsection{Krawtchouk polynomials}
\[K_n(k;p,N)= {}_2F_1\left(-n,-k;-N;\frac1p\right),\qquad 0<p<1,\quad k\in\{0,1,\ldots,N\}.\]
\begin{eqnarray*} &\w(k) = \left({N\atop k}\right)p^k(1-p)^{N-k},\qquad 0<p<1,\quad k\in\{0,1,\ldots,N\},\\
&a_n^2=np(1-p)(N+1-n),\quad b_n = p(N-n)+n(1-p).\end{eqnarray*}

\subsection{Generalized Krawtchouk polynomials}
\[\w(k) = \left({N\atop k}\right)\frac{z^k}{(1-\a)_k},\qquad k\in\{0,1,\ldots,N\}\]
\begin{eqnarray*}
(x_n+y_n)(x_{n+1}+y_n)&=-\frac{y_n(N+1+Ny_n)(N+1-\a+Ny_n)}{zN},\\
(x_n+y_n)(x_{n}+y_{n-1})&=\frac{x_n(N+1-Nx_n)(N+1-\a-Nx_n)}{N(Nx_n-n)},
\end{eqnarray*}
\begin{eqnarray*}x_0&=0,\\ y_0&=-\frac{N+1+z-\a}{N}-\frac{z}{1-\a}\,\frac{M(1-N,2-\a,-z)}{M(-N,1-\a,-z)}
\\&=-\frac{N+1+z-\a}{N}-\frac{zL_{N-1}^{(1-\a)}(-z)}{NL_{N}^{(-\a)}(-z)}.\end{eqnarray*}

\[A=\tfrac12(\a-N-1)^2,\quad B=-\tfrac12(n-N)^2,\quad C=-\ep(n+\a),\quad D=-\tfrac12\ep^2.\]

\end{document}